\documentclass[preprint,10pt]{elsarticle}

\usepackage{algorithm}
\usepackage[noend]{algpseudocode}
\usepackage{graphicx}

\usepackage{amssymb}
\usepackage{amsmath}
\usepackage{tikz}
\usepackage{textcomp}
\usepackage[utf8]{inputenc}

\journal{Elsevier}

\newenvironment{proof}[1][Proof.]{\begin{trivlist}\item[\hskip \labelsep {\bfseries #1}]}{{\hfill \ensuremath{\Box}}\end{trivlist}}

\newtheorem{theorem}{Theorem}
\newtheorem{lemma}[theorem]{Lemma}

\newtheorem{definition}{Definition}

\newcommand\NP{\ensuremath{\text{\textsc{NP}}}}
\newcommand\PTime{\ensuremath{\text{\textsc{P}}}}

\newcommand{\monitors}{M}
\newcommand{\ms}{\gamma_m}
\newcommand{\wms}{\gamma_m}
\newcommand{\MS}{\textsc{EdgeMonitoring}}
\newcommand{\WMS}{\textsc{WEM}}

\newcommand\opt{\ensuremath{\textsc{OPT}}}

\newcommand\repr{\ensuremath{\mathit{repr}}}

\newcommand{\oproblem}[3]{{%
\vspace*{0.06cm} \noindent
    \frame{%
      \frame{%
        \fbox{\begin{minipage}{0.9\textwidth}%
        \textsc{{#1}}\vspace{0.1cm}\\
        \textbf{Input:} {#2}\\%
	\textbf{Output:}{ #3}%
        \end{minipage}
	}%
      }%
    }%
}}

\newcommand{\bproblem}[3]{{%
\vspace*{0.06cm} \noindent
    \frame{%
      \frame{%
        \fbox{\begin{minipage}{0.9\textwidth}%
        \textsc{{#1}}\vspace{0.1cm}\\
        \textbf{Input:} {#2}\\%
	\textbf{Question:}{ #3}%
        \end{minipage}
	}%
      }%
    }%
}}

\begin{document}

\begin{frontmatter}

\title{Complexity of Edge Monitoring on Some Graph Classes \tnoteref{confref}}

\tnotetext[confref]{A part of this paper has been presented at Discrete Mathematics Days 2016 \cite{bagan2016edge}}

\author[GOAL]{Guillaume Bagan\corref{cor1}}
\author[GOAL]{Fairouz Beggas}
\author[GOAL]{Mohammed Haddad}
\author[GOAL]{Hamamache Kheddouci}

\address[GOAL]{Université Lyon 1, LIRIS UMR CNRS 5205, F-69621, Lyon, France}

\cortext[cor1]{Corresponding Author, guillaume.bagan@liris.cnrs.fr}

\begin{abstract}
In this paper, we study the complexity of the edge monitoring problem. 
A vertex $v$ monitors an edge $e$ if both extremities together with $v$ form a triangle in the graph. 
Given a graph $G=(V,E)$ and a weight function on edges $c$ where $c(e)$ is the number of monitors that needs the edge $e$,
the problem is to seek a minimum subset of monitors $S$ such that 
every edge $e$ in the graph is monitored by at least $c(e)$ vertices in $S$.
In this paper, we study the edge monitoring problem on several graph classes such as complete graphs, block graphs, cographs, split graphs, interval graphs and planar graphs. 
We also generalize the problem by adding weights on vertices.
\end{abstract}


\end{frontmatter}

\section{Introduction}
In this paper, we are interested in a variant of the dominating set problem: 
the edge monitoring problem.
The edge moniroring (or watchdog technique) is a mechanism for the security of wireless sensor networks  \cite{BMVH,GZYN,DXYCX}.
The basic idea is to select some 
nodes as monitors in {a given} sensor network. These monitors are employed for carrying out monitoring operations 
by promiscuously listening to the transmission of two nodes.
They can also perform basic operations of communication and sensing in the network.

The edge monitoring problem is defined as follows.
Let $G=(V,E)$ be a graph and $c$ be an integer weight function on edges of $G$. 
An edge monitoring set of $(G, c)$ is a set of vertices $S$ such that each edge $e$ of $G$ is 
monitored by at least $c(e)$ vertices of the set $S$.
A node $v \in V$ monitors an edge $e \in E$ if its both 
end-nodes are neighbors of $v$ \textit{i.e.}, $e$ together with $v$ form a triangle in the graph.
Consider the example in Figure \ref{blacknode}. The black nodes can monitor all edges depicted in bold.

\begin{figure}[htb]
\centering

\begin{tikzpicture}[scale=0.4]

\node[draw, circle, inner sep=3pt] (v1) at (0,4) {};
\node[draw, circle, inner sep=3pt] (v2) at (1,1) {};
\node[draw, circle, inner sep=3pt] (v3) at (2,0) {};
\node[draw, circle, inner sep=3pt, fill] (v4) at (2,3) {};
\node[draw, circle, inner sep=3pt] (v5) at (4,5) {};
\node[draw, circle, inner sep=3pt, fill] (v6) at (6,3.5) {};
\node[draw, circle, inner sep=3pt] (v7) at (6,0) {};	
\node[draw, circle, inner sep=3pt] (v8) at (8,3) {};
\node[draw, circle, inner sep=3pt] (v9) at (9,5) {};
\node[draw, circle, inner sep=3pt] (v10) at (10,0) {};
\node[draw, circle, inner sep=3pt] (v11) at (11,4) {};

\draw [ultra thick] (v1) -- (v2);
\draw (v2) -- (v3);
\draw (v1) -- (v4);
\draw (v2) -- (v4);
\draw (v3) -- (v4);
\draw [ultra thick] (v1) -- (v5);
\draw (v4) -- (v5);
\draw (v4) -- (v6);
\draw (v5) -- (v6);
\draw (v3) -- (v6);
\draw (v6) -- (v7);
\draw [ultra thick] (v3) -- (v7);
\draw (v6) -- (v8);
\draw (v7) -- (v8);
\draw (v6) -- (v9);
\draw [ultra thick] (v8) -- (v9);
\draw (v7) -- (v10);
\draw (v8) -- (v10);
\draw (v8) -- (v11);
\draw (v9) -- (v11);
\draw (v10) -- (v11);

\draw[->, very thick, dotted, color=blue] (v4) -- (0.5, 2.5);
\draw[->, very thick, dotted, color=blue] (v4) -- (2, 4.5);
\draw[->, very thick, dotted, color=blue] (v6) -- (4, 0);
\draw[->, very thick, dotted, color=blue] (v6) -- (8.5, 4);

\end{tikzpicture}
\caption{Edge monitoring set of a graph}
\label{blacknode}
\end{figure}
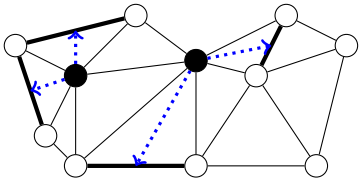

Dong et al. \cite{DXYCX} proved that the edge monitoring problem is NP-complete even restricted to unit disk graphs
and they propose a polynomial-time approximation scheme for this class of graphs.
Baste et al. \cite{parametrizedplanar2016} focused on parametrized complexity. They proved that the problem is $W[2]$-hard 
on general graphs and proposed an FPT algorithm for planar graphs and, more generally, for apex-minor-free graphs.

This paper  focuses on the complexity of the edge monitoring problem and its weighted version on different classes of graphs.
A weighted version of the edge monitoring problem is applied on graphs with weights on vertices (in addition to weights on edges).
Let $(G=(V,E),c,w))$ be a weighted graph with $w(v)$ the weight associated to a vertex $v \in V$.
The aim is to find a set $S$ that monitors $(G, c)$
and minimizes $w(S)$.

Among the classes studied in this paper, we consider block graphs, split graphs, cographs and interval graphs which are perfect graphs.
Note that the class of complete graphs is included in all graph classes mentioned before.
Since we prove that the edge monitoring problem is hard for complete graphs, we consider the problem in these 
classes with more restricted conditions.
We also have a special interest in the unit disc graphs and planar graphs.

This paper is organized as follows. Section 2 gives formal definitions of the problem and its variant.
Some basic graph terminologies and concept of complexity are also presented.
Section 3 presents some introductory results. In Section 4, we study the problem in complete graphs and block graphs.
We give a polynomial time approximation scheme for weighted complete graphs. 
Sections 5,6,7 are dedicated to interval graphs, cographs and split graphs respectively.
In section 8, we prove that the problem is NP-complete on planar unit disk graphs.
Besides, we show that there exists a PTAS for Weighted Edge Monitoring on weighted planar graphs and more generally on weighted apex-minor-free families of graphs.
The last section summarizes all results of this paper and give some suggestions for further research.

\section{Preliminaries}

In this section, we give some basic graph terminology and complexity used throughout this paper.
We also give definitions of the edge monitoring problem and all concepts used around this problem.

\subsection{Basic notions of graphs}
Graphs considered in this paper are simple, undirected and without loops.
Let $G = (V, E)$ be a graph. 
The (open) neighborhood of a vertex $v$ is $N(v) = \{ u : \{u, v\} \in E \}$.
The closed neighborhood of $v$ is $N[v] = N(v) \cup \{ v \}$.
For a set $S \subseteq V$, $N[S] = \bigcup_{v \in S}N[V]$.
The induced graph of $G$ by $S$, denoted by $G[S] = (S, E')$ contains all the edges of 
$E$ whose extremities belong to $S$.
A clique is a set $K \subseteq V$ such that each two vertices of $K$ are adjacent.
An independent set is a set $S \subseteq V$ such that no edge of $G$ has its two end vertices in $S$.
The clique number of $G$, denoted by $\omega(G)$, is the cardinality of a maximum clique in $G$. 
A graph is chordal if it has no induced cycle of length more than $3$.
The treewidth of $G$, denoted by $\mathrm{tw}(G)$, is $\min \{ \omega(H) : H \mbox{ is chordal } \wedge G 
\mbox{ is a subgraph of } H \} - 1$.
A set $S \subseteq V$ is a dominating set of $G$ if $N[S] = G$.
A set $S \subseteq V$ is a total dominating set of $G$ if $N(S) = G$.
a set $S \subseteq V $ is a double dominating set of $G$ if for every vertex $x  \in V$, $|N[x] \cap S| \geq 2$.
$\gamma(G$) (resp. $\gamma_t(G)$, $\gamma_{\times 2}(G)$) denotes the size of a  smallest dominating set (resp. total dominating set, double dominating set) of $G$ or $+\infty$ if such a set does not exist.

\subsection{Edge monitoring}
Let $e = \{v_1, v_2\}$ be an edge of a graph $G$.
We denote by $\monitors(e)$ the set of vertices $v$ such that $\{v_1, v_2, v\}$ forms a triangle. We say that $v$ monitors $e$.
Let $\alpha \geq 0$ be an integer.
A set $S \subseteq V$ $\alpha$-monitors an edge $e$ if
$|\monitors(e) \cap S| \geq \alpha$.
Let $G = (V, E)$ be a graph and $c: E \rightarrow \mathbb{N}$ be a weight function over the edges of $G$.
$S$ monitors $(G, c)$ if $S$ $c(e)$-monitors every edge $e$ in $G$.
The couple $(G, c)$ is called a weighted graph.
$\ms(G, c) = \{ |S| : S\textrm{ is a monitoring set of }(G, c) \}$ (and $+\infty$ if no monitoring set exists).
$\ms(G) = \ms(G, c)$ where $c$ is 1-uniform.



We define the problem $\MS$ as a decision problem.
However, we use the same name for the minimization problem and the parameterized version with $k$ as parameter.

\bproblem{$\MS$}
{A weighted graph $(G, c)$, an integer $k \geq 0$}
{Is there a monitoring set $S$ of $G$ such that $|S| \leq k$?}

Let $(G, w, c)$ such that $G = (V, E)$ is a graph, $w: V \rightarrow \mathbb{Q}^+$ and $c: E \rightarrow \mathbb{N}$.
$\wms(G, w, c) = \min \{ w(S) : S \textrm{ monitors } (G, c) \}$.
$(G, w, c)$ is also called a weighted graph.
Similarly to $\MS$, we define the problem $\WMS$.



\bproblem{$\WMS$}
{A weighted graph $(G, w, c)$, a number $k \in \mathbb{Q}^+$}
{Is there a monitoring set $S$ of $G, c$ such that $w(S) \leq k$?}

Let $(G, c)$ be a weighted graph with $G = (V, E)$. Then $C(G, c) = \max \{ c(e) : e \in E\}$. 
Whenever $G$ and $c$ are obvious from the context, we write $C$ instead of $C(G, c)$.
A family of weighted graphs $\mathcal{F}$ is $C$-bounded if there exists an integer $m$ 
such that $C(G, c) \leq m$ for every $(G, c) \in \mathcal{F}$.

\subsection{Complexity}

Let $X$ be a minimization problem. Let $\rho > 1$. An algorithm $A$ is called a $\rho$-approximation algorithm 
for $X$, if, for all instances $I$ of $X$, it delivers a feasible solution with objective value $A(I)$ such that
$ A(I) \leq  \rho \cdot \mbox{OPT}(I)$.
A polynomial time approximation scheme (PTAS for short) for $X$ is a family of $(1+ \epsilon)$-approximation
algorithms computable in polynomial time in the input size for any $\epsilon > 0$.

Parameterized complexity consists in studying the complexity of problems according to their input size, but also 
to another parameter. For any basic notions of parameterized complexity ($W[1]$, FPT-reduction, etc.); see \cite{Flum:2006}. 

In the folowing, we prove that 1-uniform $\MS$ cannot be approximated with a constant ratio.
We use a reduction from this problem.

\oproblem{TotalDominatingSet}
{A graph $G = (V, E)$ without isolated vertex}
{a minimum total dominating set of $G$}

\begin{theorem}\label{non-approx}
1-uniform $\MS$ cannot be approximated within $(1 - \epsilon)\ln |V|$ for any $\epsilon > 0$,
unless $\NP \subseteq \textsc{DTIME}(n^{O(\log \log n)})$.
\end{theorem}

\begin{proof}
It has been proved in \cite{chlebik2008approximation} that \textsc{TotalDominatingSet} cannot be approximated 
within $(1 - \epsilon)\ln |V|$ for any $\epsilon > 0$, unless $\NP \subseteq \textsc{DTIME}(n^{O(\log \log n)})$.
We will define an approximation preserving reduction from \textsc{TotalDominatingSet} to 1-uniform $\MS$.
Let $G=(V,E)$ be a graph without isolated vertex. We construct $G'$ from $G$ by adding three vertices $u,v,w$ which form 
a clique and connecting $u$ to every vertex in $V$. We will prove that $\ms(G')= \gamma_t(G)+3$.

Let $S$ be a total dominating set of $G$ and $S'= S \cup \{ u, v, w \}$. Then $S'$ is a monitoring set of $G'$.
Indeed, the edges $uv$, $uw$ and $vw$ are monitored by $w$, $v$ and $u$ respectively. The edges in $E$ are monitored by $u$.
Let $x$ be a vertex in $V$ then $x$ has a neighbor $y$ in $S$. Thus, $ux$ is monitored by $y$.

Now, let $S$ be a monitoring set of $G'$. $\{ u, v, w \} \subseteq S$. Otherwise, $uv$, $vw$ or $uw$ is not monitored by $S$.
Let $S'= S \setminus \{ u, v, w \}$. We will prove that $S'$ is a total dominating set of $G$.
Let $x$ be a vertex of $G$. The edge $xu$ is monitored by a vertex $y$ in $S'$. Since $\{x, y, u\}$ forms a triangle,
$x$ is adjacent to a vertex in $S'$. Hence, $\ms(G')= \gamma_t(G)+3$.

Using the same method as in Theorem 1 of \cite{klasing2004hardness} we obtain the desired result.

\end{proof} 
\section{Complete graphs and block graphs}
In this section we present some results of $\WMS$ problem on complete graphs and block graphs.

A block graph is a graph where each biconnected component (block) is a clique.
The block-cut tree $T$ of a connected graph $G$ is defined as follows. The vertices of $T$ are the blocks and the articulation points of $G$.
There is an edge between an articulation point $v$ and a block $B$ in $T$
if $v \in B$.

\begin{lemma}\label{unweightedclique1}
Let $(G, c)$ be a weighted graph such that $G = (V, E)$ is a complete graph, $C = \max\{ c(e) : e \in E \}$ and $|V| \geq C + 2$.
Then, $C \leq \ms(G, c) \leq C+2$.
Moreover, every set $S \subseteq V$ with $|S| \geq C+2$ is a monitoring set of $(G, c)$.
\end{lemma}

\begin{proof}
Since there exists an edge $e$ of weight $c(e) = C$, we need $C$ vertices to monitor it. Thus, $C \leq \ms(G)$.
Let $S \subseteq V$ be a set such that $|S| \geq C+2$. Then, every edge $e$ is $c(e)$-monitored by $S$. Indeed, let $e = \{u, v\} \in E$.
Then, the set $S \setminus \{u, v\}$ of size at least $C \geq c(e)$ $c(e)$-monitors $e$.
\end{proof}

\begin{lemma}\label{unweightedclique2}
Let $(G, c)$ be a weighted graph such that $G = (V, E)$ is a complete graph and $c$ is k-uniform with $k>0$ and $|V| \geq k + 2$.
Then, $\ms(G, c) = k+2$.
\end{lemma}

\begin{proof}Assume, for the sake of contradiction, that there exists a set $S$ that monitors $G$ such that $|S| < k+2$.
If $|S| = 1$, let $v$ be the unique element of $S$. Let $e$ an edge incident to $v$. Then, $e$ is not $c(e)$-monitored by $S$.
Otherwise, let $u$ and $v$ be two elements in $S$.
Then, $M(\{u, v\}) \cap S = |S| - 2 < k$ so $\{u, v\}$ is not monitored by $S$.
\end{proof}

\begin{theorem}\label{np-complete-on-cliques}
$\MS$ is $\NP$-complete on complete graphs. Moreover, $\MS$ is $W[1]$-complete on  complete graphs.
\end{theorem}

\begin{proof}
We will prove that $\MS$ is equivalent to \textsc{IndependentSet}
under FPT-reductions. Since  \textsc{IndependentSet} in $W[1]$-complete, the results follow.

First, we show a reduction from \textsc{IndependentSet} to $\MS$.
Let $(G = (V, E), k)$ be an instance of \textsc{IndependentSet}.
Without loss of generality, we can assume that $G$ is connected. Indeed, it is easily 
seen that \textsc{IndependentSet} remains $W[1]$-hard under this restriction.
We build an instance $(G'=(V, E'), c, k)$ of $\MS$ as follows:
$G'$ is a complete graph and for each edge $e\in E'$, we have $c(e) = k - 1$ if $e \in E$ and $c(e) = 0$ otherwise.

We show that $(G, k)$ is a positive instance of \textsc{IndependentSet} if and only if $(G',c, k)$ is a positive instance of $\MS$.
First of all, notice that there is no monitoring set of size less than $k$. Indeed, assume, for the sake of 
contradiction, that there is a monitoring set $S$ of size less than $k$.
Since $G$ is connected, there exists an edge $e$ incident to a vertex in $S$ and such that $c(e) = k-1$. We 
have $M(e) \cap S < k - 1$ so there is a contradiction.

Now, let $S \subseteq V$ such that $|S| = k$.
Then, we have:\\
$S$ is a monitoring set of $(G', c)$
iff for each $e\in E$, $|S \setminus e| \geq k - 1$
iff for each $e \in E$ in E, $|S \cap e| \leq 1$
iff $S$ is a stable of $G$.

Now, we show a Turing FPT-reduction from $\MS$ to \textsc{IndependentSet}.
The reduction is presented in Algorithm 1. Notice that this algorithm is recursive.

\begin{algorithm}
 \label{algo_w1}
\caption{}
 \begin{algorithmic}[1]
\Require $G=(V,E), c, k$
\State Let $C = \max\{ c(e) : e \in E \}$
 \If{$C > k$}
    \State \Return False
 \Else
    \If {Algorithm 1 with parameters $G, c, k-1$ returns True}
      \State \Return True
    \Else    
       \State Let $V^{*}$ built from $V$ by removing the extremities of edges $e$ with $c(e) = k$
       \State Let $E^{*}= \{ uv \in E : c(uv)=k-1 \wedge u \in V^* \wedge v \in V^* \}$
       \If{there exists an independent set of size $k$ in $G^{*}=(V^{*}, E^{*})$}
         \State \Return True
       \Else 
         \State \Return False
       \EndIf
    \EndIf 
 \EndIf 
 \end{algorithmic}
\end{algorithm}

First, let us prove that $(G, c)$ admits a monitoring set of size at most $k$ if Algorithm 1 returns True.
We proceed by induction on $k$. If $k=0$, it is clear that Algorithm 1 returns True if and only if $C=0$.
Now, assume that $k>0$. 
If Line 6 returns True then $(G,c)$ admits a monitoring set of size at most $k-1$ 
by induction hypothesis.
Assume now that Line 11 returns True. Then, there exists an independent set $S$ of size $k$ in $G^*$.
Thus, $S$ is a monitoring set of $(G,c)$. Indeed $(G,c)$ does not admit an edge $e$ with $c(e) > k$
by Lines 2-3. Edges $e$ with $c(e) = k$ have no extremities in $S$ by construction of $G^*$.
Hence, these edges are monitored by $S$.
Edges $e$ with $c(e) = k-1$ have at most one extremity in $S$ also by construction of $G^*$.
Thus, these edges are monitored by $S$.
Edges $e$ with $c(e) \leq k-2$ are necessarily monitored by $S$ since $\lvert S\rvert = k$.

Now, let us prove that Algorithm 1 returns True if $(G, c)$ admits a monitoring set $S$ of size at most $k$.
We proceed by induction on $k$. If $k=0$ then necessarily $C=0$. Thus, Algorithm 1 returns True.
Now, assume that $k>0$.
If $\lvert S \rvert \leq k-1$ then Algorithm 1 returns True in Line 6 by induction hypothesis.
Assume now that $\lvert S \rvert = k$ then it is easily seen that $S$ is an independent set of $G^*$ with $\lvert S \rvert = k$.
Then Algorithm 1 returns True in Line 11.
This completes the proof.
\end{proof}

\begin{lemma}\label{wms-cbounded-complete}
$\WMS$ can be solved in polynomial time on $C$-bounded weighted complete graphs.
\end{lemma}

\begin{proof}
Let $(G = (V, E), w, c)$ with $G$ a complete graph.
By Lemma \ref{unweightedclique1}, $\ms(G, c) \leq C+2$. Therefore, it suffices to enumerate all sets $S \subseteq V$ that monitor $G$
and such that $|S| \leq C+2$. There are $O(n^{C+2})$ such sets. Thus, the problem can be computed in polynomial time.
\end{proof}

\begin{lemma}
$\WMS$ can be solved in quasi-linear time on uniform complete graphs.
\end{lemma}

\begin{proof}
Let $(G = (V, E), w, c)$ such that $G$ is a complete graph and $c$ is $l$-uniform.
By Lemma \ref{unweightedclique2}, $\ms(G, c) = C + 2$ and by Lemma \ref{unweightedclique1}, every set $S \subseteq V$ of size $C+2$ monitors $G$.
Thus, if we choose $S$ as the set of the $C + 2$ first elements in $V$ sorted by increasing weight, we obtain an optimal solution for $\WMS(G, w, c)$.
We only need to sort $V$ which can be done in time $|V| \log |V|$.
\end{proof}

The following lemma is useful to establish the connection between $\wms$ of a graph $G$ and $\wms$ of 
its $2$-connected components.

We denote $\wms(G_1, w, c | u) = \min \{ w(S) : S \mbox{ is a monitoring set of }(G, c)\mbox{ and }u \in S\}$

\begin{lemma}\label{blockgraph-decomp}
 Let $(G = (V, E), w, c)$ be a weighted graph, $G_1=(V_1,E_1)$ and $G_2=(V_2,E_2)$ two graphs and $u\in V$ such that 
$V = V_1\cup V_2$, $E = E_1 \cup E_2$ and $V_1 \cap V_2=\{u\}$.
Let $d=\wms(G_1, w, c | u)-\wms(G_1, w, c)$.
Let $w'$ obtained from $w$ by replacing the weight of $u$ by $d$.
Then $\wms(G, w, c)=\wms(G_1, w, c)+\wms(G_2, w', c)$.
\end{lemma}

\begin{proof}
Let $S_1, S'_1, S_2$ be optimal solutions of $\WMS(G_1, w, c)$, $\WMS(G_1,w,c|u)$, $\WMS(G_2, w', c)$ respectively.

We first prove $\wms(G, w, c)\leq \wms(G_1, w, c)+\wms(G_2, w', c)$:
if $u \notin S_2$ then $S_1 \cup S_2$ is a solution of $\WMS(G, w, c)$ having weight $w(S_1)+w'(S_2)$.
If $u \notin S_2$ then $S'_1 \cup S_2$ is a solution of $\WMS(G, w, c)$ having weight $w(S'_1)+w(S_2)-d = w(S_1)+w'(S_2)$.
Thus we have $\wms(G, w, c) \leq \wms(G_1, w, c)+\wms(G_2, w', c)$.

Now we prove $\wms(G, w, c) \geq \wms(G_1, w, c)+\wms(G_2, w', c)$:
let $S^*$ be an optimal solution of $\WMS(G, w, c)$. We have $S_1^*=S^* \cap V_1$ and $S_2^*=S^* \cap V_2$
are solutions of $\WMS(G_1, w, c)$ and $\WMS(G_2, w', c)$ respectively.
We have to consider two cases:

$u \notin S^*$: 
  We have  $w(S_1^{*})\geq w(S_1)$ and $w_2'(S_2^{*})\geq w_2'(S_2)$ by optimality of $S_1$ and $S_2$.
  Since $w(S^*)=w(S_1^*)+w(S_2^*)$, $w(S^*)\geq w(S_1)+w(S_2)$.

$u \in S^{*}$ :
  This implies that $w(S^*)=w(S_1^*)+w'(S_2^*)-d$.
  Since $w'(S_2^*) \geq w(S_2)$ and $w(S_1^*)\geq w(S_1')$, then  
  $$w(S^*)\geq w(S_1')+w_2'(S_2)-d = w(S_1)+w_2'(S_2)$$ 
  Consequently we have 
  $\wms(G, w, c)\geq \wms(G_1, w, c)+\wms(G_2, w', c)$.
  This completes the proof of the lemma.
\end{proof}


\begin{theorem}
The two statements hold:
\begin{enumerate}
\item
$\WMS$ can be solved in polynomial time on $C$-bounded weighted block graphs.
\item 
$\WMS$ can be solved in quasi-linear time for block graphs $(G = (V, E), w, c)$ where $c$ is uniform.
\end{enumerate}
\end{theorem}

\begin{proof}
Without loss of generality, we can assume that $G$ is connected.
We will prove the first statement. The proof of the second statement is similar.
Let $(G = (V, E), w, c)$ be a $C$-bounded weighted block graph.
We first compute the block-cut tree $T$ of $G$. This can be done in linear time \cite{hopcroft1973algorithm}.
Then, we choose a clique $V_1$ that corresponds to a leaf of $T$ and $u$ the articulation point that is neighbor of $V_1$ in $T$.
Let $G_1 = (V_1, E_1) = G[V_1]$ and $G_2 = (V_2, E_2) = G[(V \setminus V_1) \cup \{u\}]$.
$G_2$ is also a block graph. Thus, we can apply Lemma \ref{blockgraph-decomp}.
It suffices to compute $\wms(G_1, w, c)$, $\wms(G_1, w, c | u)$ and $\wms(G_2, w', c)$.
$\wms(G_1, w, c)$ can be computed in polynomial time by using Lemma \ref{wms-cbounded-complete}.
Proof of Lemma \ref{wms-cbounded-complete} can be easily modified to compute $\wms(G_1, w, c | u)$.
$\wms(G_2, w', c)$ can be computed by induction.
\end{proof}

\section{PTAS for the $\WMS$ problem in weighted complete graphs}
In this section, we study the approximation complexity of the weighted monitoring set 
problem in vertex-weighted complete graphs.\\

\begin{theorem}
There exists a PTAS for $\WMS$ on complete graphs.
\end{theorem}

\begin{proof}
 Fix $\epsilon > 0$ and $k=\lceil2/\epsilon\rceil$.
 Let $G=(V,E),w,c$ such that $G$ is a complete graph and $C = \max\{ c(e) : e \in E \}$.
 Let $\opt$ denote an optimal solution for $\WMS(G,w, c)$.

 We have to consider three different cases:\\
 \textbf{Case 1.} $C\leq k$ : \\
 Using Lemma \ref{unweightedclique1}, we have $|\opt|\leq C+2\leq k+2$.
 We just need to enumerate all the sets with size at most $k+2$. We can do it in polynomial time $O(n^{k+2})$.\\
 \textbf{Case 2.} $|V| < C+2$: \\
 Clearly, there exists no monitoring set for $(G, c)$ since there exists an edge $e = \{u, v\}$ such that $c(e) = C$ and $\monitors(e) < C$.\\
  \textbf{Case 3.} $C \geq k$ and $|V| \geq C+2$: \\
 Let $S_{first}$ be the set of the first $C+2$ vertices sorted in ascending order by weight $w(v)$.
 Let $\mathcal{C}$ be the set of sets $S\subseteq V$ such that $C\leq |S| \leq C+2$ and $|S \setminus S_{first}| \leq k$.
 We prove that $\mathcal{C}$ has a polynomial size. Indeed,  we have\\
$$|\mathcal{C}| \leq |\{ S \cap S_{first} : S \in \mathcal{C} \})| \times |\{ S \setminus S_{first} : S \in \mathcal{C} \})|$$
It holds

 \begin{equation}
    \begin{aligned}
|\{ S \cap S_{first} : S \in \mathcal{C} \})| = \sum_{i = C}^{C+2} \sum_{j = 0}^k |\{ S \cap S_{first} : S \in V \wedge |S| = i \wedge |S \setminus S_{first}| = j\})|\\
\end{aligned}
\end{equation}

\begin{equation}
    \begin{aligned}
      = \sum_{i = C}^{C+2} \sum_{j = 0}^k \binom{C+2} {i - j} \leq O(C^k)\\
\end{aligned}
\end{equation}

Since $|S \setminus S_{first}| \leq k$ for every $S \in \mathcal{C}$, it holds $|\{ S \setminus S_{first} : S \in \mathcal{C} \})| \leq O(n^k)$.
Thus $|\mathcal{C}| = O(C ^ k n^k)$ is polynomial in $|V|$.
 
The algorithm consists to enumerate all the sets in $\mathcal{C}$ and take a solution of minimum weight.
This can be done in polynomial time. We distinguish two subcases as follows: \\
 \textbf{Case 3.a.} $\opt \in \mathcal{C}$ : \\
 Clearly, the algorithm returns an optimal solution.\\
 \textbf{Case 3.b.} $\opt \notin \mathcal{C}$ : \\
Notice that $S_{first}$ is a (non necessary optimal) solution by Lemma \ref{unweightedclique1} and the algorithm returns a solution $S$ such that $w(S) \leq w(S_{first})$.    
We will prove that $w(S_{first}) \leq (1 + \epsilon) w(\opt)$.
Let $a_1,...,a_l$ denote the vertices in $\opt \cap S_{first}$.
Let $b_1,...,b_m$ denote the vertices in $\opt \setminus S_{first}$.
Let $c_1,...,c_n$ denote the vertices in $S_{first} \setminus \opt$ sorted in ascending order by weight $w(v)$.
Since $|\opt \setminus S_{first}| \geq k$, we have $m \geq k$.
 In the following, we will bound the approximation ratio of the solution:
 
 \begin{equation} \label{approx1}
    \begin{aligned}
\frac{w(S_{first})}{w(\opt)}= \frac{w(a_1)+...+w(a_l)+w(c_1)+...+w(c_n)}{w(a_1)+...+w(a_l)+w(b_1)+...+w(b_m)} \\
\end{aligned}
\end{equation}

 \begin{equation} \label{approx2}
    \begin{aligned}
\leq \frac{w(c_1)+...+w(c_n)}{w(b_1)+...+w(b_m)}  \\
\end{aligned}
\end{equation}

 \begin{equation} \label{approx3}
    \begin{aligned}
\leq \frac{n.w(c_n)}{m.w(c_n)}= \frac{n}{m}\\
\end{aligned}
\end{equation}

 \begin{equation} \label{approx4}
    \begin{aligned}
 \leq \frac{m+2}{m} \\
\end{aligned}
\end{equation}

 \begin{equation} \label{approx5}
    \begin{aligned}
 \leq \frac{k+2}{k}=1+\frac{2}{k} \\
\end{aligned}
\end{equation}

 \begin{equation} \label{approx6}
    \begin{aligned}
\leq  1+\frac{2}{\lceil \frac{2}{\epsilon} \rceil} \leq 1+\epsilon \\
\end{aligned}
\end{equation}

In (\ref{approx1}), we use the fact that if $a,b,c>0$ and $b\geq c$ 
then $ \frac{a+b}{a+c}\leq \frac{b}{c} $. Since $w(\opt)\leq w(S_{first})$, we obtain
$w(c_1)+...+w(c_n) \geq w(b_1)+...+w(b_m)$.
Thus, we get (\ref{approx2}).
We obtain (\ref{approx3}) since $w(c_i)\leq w(c_n)$ for any $i \in [1, n]$ and $w(b_i) \geq w(c_n)$ for any $i \in [1, m]$.
To get (\ref{approx4}), we use the property that $|\opt|\geq C$ and $|S_{first}|=C+2$.
Since $m\geq k$ and $k=\lceil 2/\epsilon \rceil$, the rest follows.
 \end{proof}

\section{Interval graphs}
In this section, we give a polynomial algorithm for computing $\WMS$ on weighted interval graphs. This algorithm uses dynamic programming.
First, we introduce some definitions.

A graph $G = (V, E)$ is an interval graph if there exists $|V|$ intervals $(I_i)_{i \in V} = ([a_i, b_i])_{i \in V}$ of the real line such that
$\{i, j\} \in E$ if and only if $I_i \cap I_j \neq \emptyset$ for every distinct vertices $i, j \in V$.
We say that $(I_i)_{i \in V}$ is a realization of $G$.
Without loss of generality, we can assume that there are no intervals $I_i$ and $I_j$ that have a common extremity.

Given an interval graph $G=(V, E)$ and a realization $(I_i)_{i \in V}$,
we define a total order $<_L$ (resp. $<_R$) over $V$ such that $i <_L j$ (resp. $i <_R j$) if $a_i < a_j$ (resp. $b_i < b_j$).

The following definition is a refinement of the nice tree decomposition introduced by Kloks \cite{kloks94}

\begin{definition}\cite{nicepath-decomposition}
Let $G = (V, E)$ be an interval graph and $(I_i)_{i \in V}$ be a realization of $G$.
A nice path decomposition of $G$ is a sequence of sets of vertices $B_0, \ldots B_l$ such that
\begin{itemize}
\item all sets $B_i$ are cliques of $G$;
\item every edge $e \in E$ appears in a set $B_i$,
\item for every vertex $v \in E$, the set of indices $i$ such that $v \in B_i$ is a segment of $[0, l]$.
\item $B_0 = \emptyset$ and $B_l = \emptyset$;
\item For every $i \in [1, l]$,
\begin{itemize}
\item
$B_{i} = B_{i-1} \cup \{v\}$ ($i$ introduces the vertex $v$)
\item
or $B_{i} = B_{i-1} \setminus \{ v \}$ ($i$ forgets the vertex $v$).
\end{itemize}
\item the order in which vertices are introduced corresponds to $<_L$
\item the order in which vertices are forgotten corresponds to $<_R$
\end{itemize}
\end{definition}

\begin{lemma}\cite{nicepath-decomposition}
Let $G = (V, E)$ be an interval graph and $(I_i)_{i \in V}$ be a realization of $G$.
Then $G$ has a nice path-decomposition that can be computed in linear time.
\end{lemma}

For the next lemmas, we consider an interval graph $G = (V, E)$
and a nice path-decomposition $B_0, \ldots, B_l$ of $G$.
Moreover, we introduce the following notations.
For $i \in [0, l]$, $F_i$ is the set of vertices appearing in some set $B_j$, $j < i$, but not in $B_i$.
$V_i = F_i \cup B_i$ and $G_i = G[V_i]$.

A set $S \subseteq V_i$ is an $i$-partial solution if every edge $e$ in $G_i$ that has an extremity in $F_i$ is $c(e)$-monitored by $S$.
The $i$-representant $W$ of $S \subseteq V_i$, denoted by $\repr_i(S)$,
contains exactly the $C+2$ greatest vertices in $S \cap N[B_i]$ w.r.t. $<_R$ or is $S \cap N[B_i]$ if $|S \cap N[B_i]| < C+2$.
We say that $S$ extends $W$ if $W$ is the $i$-representant of $S$.

We denote by $\mathcal{F}^*_i$ the set of $i$-representants of
$i$-partial solutions.
$w_i^*: \mathcal{F}^*_i \rightarrow \mathbb{Q}^+$ 
is a function  such that $w_i^*(W) = \min\{ w(S) : S \textrm{ is an }i\text{-partial solution that extends }W\}$.

Before presenting the algorithm, we introduce two lemmas. The second is the key of the algorithm.

\begin{lemma}\label{interval-key-lemma1}
Let $u \in B_i$, $v_1, v_2 \in V_i$ such that $v_1 <_R v_2$ and $v_1 \in N[u]$.
Then $v_2 \in N[u]$.
\end{lemma}

\begin{proof}
Let $[a_u, b_u]$, $[a_{v_1}, b_{v_1}]$ and $[a_{v_2}, b_{v_2}]$ the intervals that represent $u$, $v_1$ and $v_2$ respectively in the realization
of $G$.
Since $u \in B_i$ and $v_2 \in V_i$, $b_u > a_{v_1}$ and $b_u > a_{v_2}$.
Since $v_1 \in N[u]$, we have $a_u < b_{v_1}$ and since $v_1 <_R v_2$, we have $a_u < b_{v_2}$.
Thus $[a_u, b_u] \cap [a_{v_2}, b_{v_2}] \neq \emptyset$. Consequently, $v_2 \in N[u]$.
\end{proof}

\begin{lemma}\label{interval-key-lemma2}
Let $S \subseteq V_i$, $W = \repr_i(S)$, $v_1, v_2 \in B_i$ and $\alpha \in [0, C]$.
If $\{v_1, v_2\}$ is $\alpha$-monitored by $S$
then $\{v_1, v_2\}$ is $\alpha$-monitored by $W$.
\end{lemma}

\begin{proof}
First, notice that $\monitors(\{v_1, v_2\}) \subseteq N[B_i]$.
If $|S \cap N[B_i]| \leq C+2$, then $W = S \cap N[B_i]$ and the lemma is trivially verified.
Now, assume that $|S \cap N[B_i]| > C+2$ and let $u \in (S \setminus W)  \cap \monitors(\{v_1, v_2\})$.
By Lemma \ref{interval-key-lemma1}, every vertex $u' \in W$ belongs to $N[v_1]$ and $N[v_2]$.
So all elements in $W$ except at most two ($v_1$ and $v_2$) belong to $\monitors(\{v_1, v_2\})$.
Thus $|\monitors(\{v_1, v_2\}) \cap W| \geq C$ and $\{v_1, v_2\}$ is $\alpha$-monitored by $W$.
\end{proof}

To solve $\WMS$ on interval graphs, a naive algorithm consists to iterate over the sets $B_i$ and to compute for each $i$ the set of $i$-partial solutions.
Unfortunately, the algorithm is non polynomial since the set of $i$-partial solutions can be exponential.
The key of the algorithm is as follows: instead of considering all the $i$-partial solutions,
we consider the representants of the $i$-partial solutions.
Since the number of representants is polynomially bounded by $|V|$, the algorithm will run in polynomial time.
Lemma \ref{interval-key-lemma2} guarantees that we don't miss solutions.
Indeed, let $S$ be an $i$-partial solution. If $i+1$ introduces the node $v$, then $S$ and $S \cup \{v\}$ are $(i+1)$-partial solutions.
If $i+1$ forgets the node $v$ then $S$ is an $(i+1)$-partial solution if and only if every forgotten edge $e$ (i.e. an edge having
$v$ as extremity and the other extremity in $B_{i+1}$) is $c(e)$-monitored by $S$.
But thanks to Lemma \ref{interval-key-lemma2}, it suffices to check that these edges are $c(e)$-monitored by $\repr_i(S)$.



We present now Algorithm 2. 

\begin{algorithm}
 \label{algo_intervalgraphs_update}
\caption{Algorithm for $\WMS$ on interval graphs}
 \begin{algorithmic}[1]
\Require {a weighted interval graph $(G, c, w)$ and a nice path decomposition $B_0, \ldots, B_l$}
\State $\mathcal{F}_0 \leftarrow \{ \emptyset \}$
\State $w_0(\emptyset) = 0$
\For{$i$ from $1$ to $l$}
 \State $\mathcal{F}_i \leftarrow \emptyset$
  \State $w_i(S) = +\infty$ for any $S$
 \If{$i$ forgets the node $v$}
    \For{$W \in \mathcal{F}_{i-1}$}
      \If{every edge $e = \{u, v\}$ with $u \in B_i$ is $c(e)$-monitored by $W$}
         \State $W' \leftarrow \repr_i(W)$
         \State $\mathcal{F}_i \leftarrow \mathcal{F}_i \cup \{W'\}$
         \State $w_i(W') \leftarrow \min \{w_i(W'), w_{i-1}(W) \}$
      \EndIf
   \EndFor
  \ElsIf{$i$ introduces the node $v$}
    \For{$W \in \mathcal{F}_{i-1}$}
       \State $W' \leftarrow \repr_i(W)$
       \State $\mathcal{F}_i \leftarrow \mathcal{F}_i \cup \{W'\}$
       \State $w_i(W') \leftarrow \min\{w_i(W'), w_{i-1}(W)\}$
       \State $W' \leftarrow \repr_i(W \cup \{v\})$
       \State $\mathcal{F}_i \leftarrow \mathcal{F}_i \cup \{W'\}$
       \State $w_i(W') \leftarrow \min\{w_i(W'), w_{i-1}(W) + w(v)\}$
    \EndFor
 \EndIf
 \EndFor
 \If{ $\mathcal{F}_l = \emptyset$}  
 \State \Return $+\infty$
 \Else
 \State     \Return $\min\{ w_l(W) : W \in \mathcal{F}_l \}$
 \EndIf
 \end{algorithmic}
\end{algorithm}

The next lemma shows that the sets $\mathcal{F}_i$ and functions $w_i$ computed by Algorithm 2 correspond to the sets $\mathcal{F}^*_i$ and functions $w^*_i$  defined previously.

\begin{lemma}\label{interval-F-and-w}
For every $i \in [0, l]$, after the run of Algorithm 2, we have
$\mathcal{F}_i = \mathcal{F}^*_i$ and $w_i(S) = w_i^*(S)$ for every $S \in \mathcal{F}_i$.
\end{lemma}

\begin{proof}
We prove by induction on $i$. The property is clearly verified for $i = 0$.
Now, suppose that the property holds for $i$ and prove it for $i+1$.

$\mathcal{F}_{i+1} \subseteq \mathcal{F}^*_{i+1}$
and for each $W \in \mathcal{F}_{i+1}$, $w^*_{i+1}(W) \leq w_{i+1}(W)$:
let $W' \in \mathcal{F}_{i+1}$. We consider two cases.

$i+1$ forgets the vertex $v$:
then, $W'$ comes from some $W \in \mathcal{F}_{i}$ such that $W' = \repr(W)$, $w_{i+1}(W') = w_i(W)$
and $W'$ is added to $\mathcal{F}_{i+1}$ by Lines 9-11.
Using the induction hypothesis, $W' \in \mathcal{F}^*_i$ and $w_i(W) = w^*_i(W)$.
Let $S$ be a $i$-partial solution of weight $w(S) = w^*_i(W)$ that extends $W$.
By Line 8 of the algorithm, every edge $e = \{u, v\}$ where $u \in B_i$ is $c(e)$-monitored by $W$ and thus by $S$.
Consequently, $S$ is an $(i+1)$-partial solution with $\repr_{i+1}(S) = \repr_{i+1}(W) = W'$.
Thus, $W' \in \mathcal{F}^*_{i+1}$ and $w^*_{i+1}(W') \leq w(S) = w^*_i(W) = w_i(W) = w_{i+1}(W')$.

$i+1$ introduces the vertex $v$:
There are two possibilities.

$v \notin W'$: then $W'$ comes from some $W \in \mathcal{F}_i$ such that $W' = \repr(W)$, $w_{i+1}(W') = w_i(W)$
and $W'$ is added to $\mathcal{F}_{i+1}$ by Lines 14-16.
By induction hypothesis, $W \in \mathcal{F}^*_i$ and $w_i(W) = w^*_i(W)$.
Let $S$ be a $i$-partial solution of weight $w(S) = w^*_i(W)$ that extends $W$.
$S$ is an $(i+1)$-partial solution with $\repr_{i+1}(S) = \repr_{i+1}(W) = W'$.
Thus $W' \in \mathcal{F}^*_{i+1}$ and $w^*_{i+1}(W') \leq w(S) = w^*_i(W) = w_i(W) = w_{i+1}(W')$.

$v \in W$: $W'$ comes from some $W \in \mathcal{F}_i$ such that $W' = \repr(W+\{v\})$, $w_{i+1}(W') = w_i(W)+w(v)$
and $W'$ is added to $\mathcal{F}_{i+1}$ by Lines 17-19.
Let $S$ be a $i$-partial solution of weight $w(S) = w^*_i(W)$ that extends $W$.
$S' = S \cup \{v\}$ is an $(i+1)$-partial solution with $\repr_{i+1}(S \cup \{v\}) = \repr_{i+1}(W \cup \{ v \}) = W'$.
Thus $W' \in \mathcal{F}^*_{i+1}$ and $w^*_{i+1}(W') \leq w(S+\{v\}) = w^*_i(W) + w(v) = w_i(W) + w(v) = w_{i+1}(W')$.

$\mathcal{F}^*_{i+1} \subseteq \mathcal{F}_{i+1}$ and for each $W \in \mathcal{F}^*_{i+1}$, $w_{i+1}(W') \leq w^*_{i+1}(W')$:
let $W' \in \mathcal{F}^*_{i+1}$ and $S'$ be an $(i+1)$-partial solution that extends $W$ and such that $w(S') = w^*_{i+1}(W')$.
We also consider two cases:

$i+1$ forgets the vertex $v$:
then $S'$ is an $i$-partial solution. Let $W = \repr_{i}(S')$. Then $W' = \repr_{i+1}(W)$.
Using the induction hypothesis, $W \in \mathcal{F}_i$ and $w_i(W) = w^*_i(W)$.
By definition of a $(i+1)$-partial solution, every edge $e = \{u, v\}$ with $u \in B_i$ is $c(e)$-monitored by $S'$.
By applying Lemma \ref{interval-key-lemma2}, these edges are also $c(e)$-monitored by $W$.
Thus, Line 8 of the algorithm succeeds and $W'=\repr_{i+1}(W)$ is added to $\mathcal{F}_{i+1}$ and by Line 10
$w_{i+1}(W') \leq w_i(W) = w^*_i(W) = w(S') = w^*_{i+1}(W')$.

$i+1$ introduces the vertex $v$.
There are two possibilities.

$v \notin S'$: 
then $S'$ is an $i$-partial solution. Let $W = \repr_{i}(S')$.
Using the induction hypothesis, $W \in \mathcal{F}_i$ and $w_i(W) = w^*_i(W)$.
Thus, $W'=\repr_{i+1}(W)$ is added to $\mathcal{F}_{i+1}$ by Line 14 and by Line 15
$w_{i+1}(W') \leq w_i(W) = w^*_i(W) = w(S') = w^*_{i+1}(W')$.

$v \in S'$:
let $S = S' - v$. Then $S'$ is an $i$-partial solution.
Let $W = \repr_{i}(S)$.
Using the induction hypothesis, $W \in \mathcal{F}_i$ and $w_i(W) = w^*_i(W)$.
Thus, $W'=\repr_{i+1}(S\cup\{v\}) = \repr_{i+1}(W\cup\{v\})$ is added to $\mathcal{F}_{i+1}$ by Line 18-19
and $w_{i+1}(W') \leq w_i(W)+w(v) = w^*_i(W)+w(v) = w(S)+w(v) = w(S') = w^*_{i+1}(W')$.
\end{proof}

\begin{theorem}\label{running-time-theorem}
$\WMS$ on $C$-bounded weighted interval graphs is in $\PTime$.
More precisely, it can be solved in time $O(|V|^{C+4})$.
\end{theorem}

\begin{proof}
Thanks to Lemma \ref{interval-F-and-w}, it is clear that Algorithm 2 is exact.
Let prove that it runs in the expected time.
The algorithm consists of a main loop that does $|V|+1$ iterations.
Within this loop, we have two possibilities: forgetting or introducing a vertex.
In the two cases, we loop over the elements of $\mathcal{F}_{i-1}$. Each step of the loop can be done
in time $O(N(B_i))$ (since $C$ is bounded) in both cases.
The size of $\mathcal{F}_{i-1}$ is bounded by $(N[B_{i-1}] \cap V_{i-1})^{C+2}$.
Therefore the time spent within a step of the main loop is $O((N[B_{i-1}]\cap V_{i-1})^{C+3})$.
Since $N[B_{i-1}] \cap V_{i-1}$ is bounded by $|V|$, Algorithm \ref{algo_intervalgraphs_update} 
runs in time $O(|V|^{C+4})$.
\end{proof}

The complexity of the algorithm can be refined in the case of unit interval graphs.

\begin{lemma}\label{interval-key-lemma3}
Let $C$ be a clique of an unit interval graph $G = (V, E)$. Then $N[C] \leq 3 \omega(G)$.
\end{lemma}

\begin{proof}
Let $(I_i)_{i \in E}$ be a realization of $G$. Since $G$ is an unit interval graph, we have
$u \leq_L v \Leftrightarrow u \leq_R v$ for every $x, y \in V$.
For every vertex $v \in V$, we denote by $N_{\leq}[v]$ (resp. $N_{\geq}[v]$) the set $\{u : u \in N[v] \wedge u \leq_L v\}$
(resp. $\{u : u \in N[v] \wedge u \geq_L v\}$).
Let $v_{min}$ (resp. $v_{max}$) be the minimal (resp. maximal) vertex of $C$ w.r.t $\leq_L$.
It is easily seen that $N[C] = N_{\leq}[v_{min}] \cup (N_{\geq}[v_{min}] \cap N_{\leq}[v_{max}]) \cup N_{\geq}[v_{max}]$ and
that $N_{\leq}[v_{min}]$, $N_{\geq}[v_{min}] \cap N_{\leq}[v_{max}]$ and $N_{\geq}[v_{max}]$ are clique of $G$.
Thus $N[C] \leq 3 \omega(G)$.9

\end{proof}

\begin{theorem}
$\WMS$ can be solved in time $O(\omega(G)^{C+3}|V|)$ on $C$-bounded weighted unit interval graphs.
\end{theorem}

\begin{proof}
We refine the running time of Theorem \ref{running-time-theorem}.
Thanks to Lemma \ref{interval-key-lemma3}, we can bound $N[B_{i-1}]\cap V_{i-1}$ by $3 \omega(G)$.
Thus, we deduce that the overall running time is $O(\omega(G)^{C+3}|V|)$ in weighted unit interval graphs. 
\end{proof}

\section{Cographs}
Let $G_1 = (V_1, E_2)$ and $G_2 = (V_2, E_2)$ such that $V_1 \cap V_2 = \emptyset$.
The join of $G_1$ and $G_2$ is the graph $G = (V_1 \cup V_2, E_1 \cup E_2 \cup \{ \{ u, v \} : u \in V_1 \wedge v \in V_2\})$.
The class of cographs is defined by induction.
\begin{itemize}
\item The graph which contains one vertex is a cograph;
\item The (disjoint) union and the join of two cographs are cographs.
\end{itemize}

\begin{lemma}\label{cographs-1}
Let $G = (V, E)$ be the join of two graphs $G_1 = (V_1, E_1)$ and $G_2 = (V_2, E_2)$.
Let $S$ be a total dominating set of $G_1$. Then, $S$ monitors all edges between $V_1$ and $V_2$.
\end{lemma}

\begin{proof}
Let $\{u, v\}$ be an edge between $G_1$ and $G_2$ such that $u \in V_1$. Then there exists a vertex $u_1 \in S$ adjacent to $u$.
Thus,  $\{u, v\}$ is monitored by $S$ since $\{u, v, u_1\}$ is a triangle of $G$.
\end{proof}

\begin{lemma}\label{cographs-2}
Let $G = (V, E)$ be the join of two graphs $G_1 = (V_1, E_1)$ and $G_2 = (V_2, E_2)$.
Let $S$ be a monitoring set of $G$.
Then $S \cap V_1$ is a total dominating set of $G_1$ or $S \cap V_2$ is a total dominating set of $G_2$.
\end{lemma}

\begin{proof}
Assume for the sake of contradiction that $S_1$ is not a total dominating set of $G_1$ and $S_2$ is not a total dominating set of $G_2$.
Then there exists an edge $\{u, v\} \in E$ such that $u$ has no neighbor in $S_1$ and $v$ has no neighbor in $S_2$.
Thus, $\{u, v\}$ is not monitored by $S$.
\end{proof}

\begin{lemma}\label{cographs-3}
Let $G$ be the join of two graphs $G_1 = (V_1, E_1)$ and $G_2 = (V_2, E_2)$.
Let $S$ be a minimal monitoring set of $G$.
Then $|S \cap V_1| \leq 1$ or $|S \cap V_2| \leq 1$.
\end{lemma}

\begin{proof}
Let $S$ be a minimal monitoring set of $G$, $S_1 = S \cap V_1$ and $S_2 = S \cap V_2$.
Assume, for the sake of contradiction, that $|S_1| \geq 2$ and $|S_2| \geq 2$.
By Lemma \ref{cographs-2}, $S_1$ is a total dominating set of $G_1$ or $S_2$ is a total dominating set of $G_2$.
By symmetry, suppose that $S_1$ is a total dominating set of $G_1$. Then $S_1$ monitors all edges between $V_1$ and $V_2$ by Lemma \ref{cographs-1} and all edges in $V_2$.
Consequently, for every vertex $u \in V_2$, $S_1 \cup \{ u \}$ is a monitoring set of $G$ since $u$ monitors all edges in $V_1$.
Thus, $S$ is not minimal.
\end{proof}

\begin{lemma}\label{cographs-4}
Let $G = (V, E)$ be a graph with no isolated vertices and $S$ a monitoring set of $G$. Then, $S$ is a total dominating set of $G$.
\end{lemma}

\begin{proof}
Let $v$ be a vertex in $V$. Since $G$ has no isolated vertices, there is a vertex $e = (v, v_1)$ incident to $v$. Since $S$ is a monitoring set of $G$,
there is a vertex $v_2 \in S$ such that $\{v, v_1, v_2\}$ is a triangle in $G$. Thus, $v$ is adjacent to a vertex in $S$.
\end{proof}

Combining Lemmas \ref{cographs-1}, \ref{cographs-2} and \ref{cographs-4}, we obtain the following lemma.

\begin{lemma}\label{cographs-5}
Let $G = (V, E)$ be the join of two graphs $G_1 = (V_1, E_1)$ and $G_2 = (V_2, E_2)$.
Let $S_1 = S \cap V_1$ and $S_2 = S \cap V_2$.
The two statements hold.
\begin{itemize}
\item If $S_1 \neq \emptyset$ and $S_2 \neq \emptyset$, then $S$ is a monitoring set of $G$ if and only if $S_1$ is a total dominating set of $G_1$ or $S_2$ is a total dominating set of $G_2$.
\item If $S_2 = \emptyset$ (resp. $S_1 = \emptyset$),
then $S$ is a monitoring set of $G$ if and only if $G_1$ (resp. $G_2$)
has no isolated vertices and $S_1$ (resp. $S_2$) is a monitoring set
of $G_1$ (resp. $G_2$).
\end{itemize}
\end{lemma}

The following lemma is a direct consequence of Lemma \ref{cographs-3} and Lemma \ref{cographs-5}.

\begin{lemma}\label{cographs-6}
Let $G = (V, E)$ be a graph.
If $G$ is the (disjoint) union of two graphs $G_1$ and $G_2$.
Then, $$\wms(G, w) = \wms(G_1, w) + \wms(G_2, w)$$

If $G$ is the join of two graphs $G_1$ and $G_2$.

$$\wms(G,w) = \min \left\{
    \begin{array}{ll}
        \gamma_t(G_1, w) +  \min\{ w(v) : v \in V_2\}\\
        \min\{ w(v) : v \in V_1\} + \gamma_t(G_2, w)\\
        \wms(G_1, w) \mbox { if } G_1 \mbox { has no isolated vertices}\\
        \wms(G_2, w) \mbox { if } G_2 \mbox { has no isolated vertices}\\
    \end{array}
\right.$$
\end{lemma}

Lemma \ref{cographs-6} combined with the fact that a cotree is computable in linear time \cite{cograph-habib} give us a linear time algorithm to compute 1-uniform $\WMS$
on cographs.

\begin{theorem}
1-uniform $\WMS$ can be solved in linear time on cographs.
\end{theorem}
\section{Split graphs and comparability graphs}

A  graph $G = (V, E)$ is a split graph is $V$ can be partionned into $C$ and $I$ where $C$ is a clique of $G$ and $I$ is an independant set of $G$.

\begin{lemma}\label{splitgraphs-lemma1}
Let $G = (V = C \cup I, E)$ be a split graph with minimum degree $\delta(G) \geq 2$ and such that $|C| \geq 3$ and 
Then, there exists a minimum 2-tuple dominating set (resp. monitoring set) $S \subseteq C$.
\end{lemma}

\begin{proof}
Let $S$ be a set that minimizes $|S \cap I|$ among all minimum 2-tuple dominating sets of $G$.
For the sake of contradiction, suppose $S \cap I$ non empty and let $v$ be a vertex in $S \cap I$.
If $N(v) \subseteq S$, then $S - v$ is also a 2-tuple dominating set of $G$. Thus, $G$ is not minimum.
Otherwise, let $u \in S \setminus N(v)$. Then $S' = S \cup \{u\} - v$ is a minimum 2-tuple dominating set of $G$ with
$|S' \cap I| < |S \cap I|$. Thus $S$ does not minimize $|S \cap I|$.

The proof for monitoring sets is quite similar to the proof for 2-tuple dominating sets.
Let $S$ be a set that minimizes $|S \cap I|$ among all minimum monitoring sets of $G$.
For the sake of contradiction, suppose $S \cap I$ non empty and let $v$ be a vertex in $S \cap I$.
$S$ contains at least 3 vertices and $|S \cap C| \geq 2$. Otherwise, $S$ does not monitor all vertices between $C$ and $I$.
If $N(v) \subseteq S$ and $|S \cap C| \geq 3$, then $S - v$ is also a monitoring set of $G$. Thus $S$ is not minimum.
If $N(v) \subseteq S$ and $|S \cap C| = 2$ then
choose a vertex $u \in C \setminus S$. Thus, $S' = S \cup \{u\} - v$ is a minimum monitoring set with $|S' \cap I| < |S \cap I|$.
Now, suppose that $N(v) \nsubseteq S$ and let $u \in N(V) \setminus S$. Then, $S' = S \cup \{u\} - v$ is a minimum monitoring set with $|S' \cap I| < |S \cap I|$.
That contradicts our assumption.
\end{proof}

\begin{lemma}
Let $G = (V = C \cup I, E)$ be a split graph with minimum degree $\delta(G) \geq 2$ and such that $|C| \geq 3$ and $\gamma_{\times 2}(G) \geq 3$.
Then, $\ms(G) = \gamma_{\times 2}(G)$.
\end{lemma}

\begin{proof}
To see that $\gamma_{\times 2}(G) \leq \ms(G)$, consider 
a monitoring set $S$ and a vertex $v$. Since $\delta(G) \geq 2$, $v$ admits a neighbor $v_1$.
Since $\{v, v_1\}$ is monitored by $S$, $v$ admits a neighbor $v_2 \in S$
and, since $\{v, v_2\}$ is monitored by $S$, $v$ admits another neighbor $v_3 \in S$. Thus $S$ is a double dominating set of $G$.

We will prove that $\gamma_{\times 2}(G) \geq \ms(G)$.
Let $S$ be a minimum 2-tuple dominating set of $G$.
Thanks to Lemma \ref{splitgraphs-lemma1}, we can assume without loss of generality that $S \subseteq C$.
Since $|S| \geq 3$, $S$ monitors all edges in $G[C]$. Let $\{u, v\}$ be an edge in $G$ such that $u \in C$ and $v \in I$.
Since $S$ dominates twice the vertex $v$, there is a node $u' \in S \cap N(v)$ distinct to $u$. Thus $\{u, v\}$ is monitored by $u'$.
Consequently, $S$ is a monitoring set of $G$.
\end{proof}

Since 2-tuple domination is $\NP$-complete on split graphs even with these restrictions \cite{double-domination-in-split-graphs}, we obtain the following result.

\begin{theorem}
1-uniform $\MS$ is $\NP$-complete on split graphs.
\end{theorem}

A graph $G = (V, E)$ is a comparability graph if there exists a poset $\leq$ over $V$ such that $\{x, y\} \in E$ if and only if $x \leq y$ or $y \leq x$ for every $x, y \in E$.

\begin{theorem}
1-uniform $\MS$ is $\NP$-complete on comparability graphs.
\end{theorem}

\begin{proof}
We do a reduction from \textsc{TotalDominatingSet} on bipartite graphs which has been proved $\NP$-complete \cite{totaldomination-bipartite}.
Let $G = (V, E)$ be a bipartite graph.
Without loss of generalility, assume that $G$ has no isolated vertices. 
Let $G'$ be the graph obtained from $G$ by adding an universal vertex $u$.
It is clear that $G'$ is a comparability graph.
We will prove that $\ms(G') = \gamma_t(G) + 1$.
Let $S$ be a total dominating set of $G$. Then, $S \cup \{u\}$ is a monitoring set of $G$. Indeed, every edge in $E$ is covered by $u$ and
for every edge $\{u, v\}$ with $v \in V$, there is a vertex $v' \in N(v) \cap S$. Thus, $\{u, v\}$ is monitored by $v'$.
Now, let $S$ be a monitoring set of $G'$. Then, $u \in S$ because $u$ is the only vertex that monitors edges in $E$.
$S - u$ is a total dominating set of $G$. Indeed, let $v$ be a vertex in $V$. $\{u, v\}$ is an edge of $G'$ monitored by a vertex $v' \in S - u$ distinct from $v$.
Thus, $v$ is dominated $v'$.
\end{proof}

\section{Planar graphs and unit disk graphs}

\subsection{Negative results}

A graph $G = (V, E)$ is an unit disk graph if it there exists a map $f: V \rightarrow \mathbb{R}^2$ satisfying
$$\{u, v\} \in E \Leftrightarrow \|f(u) - f(v)\| \leq 2$$
$f$ is called a geometric representation of $G$.

Recognizing whether a graph $G$ is an unit disk graph is $\NP$-hard \cite{unitdisk-np-complete}.
Thus, computing a geometric representation of an unit disk graph is also $\NP$-hard.
Consequently, we suppose that an unit disk graph $G$ is given with a geometric representation $f$.

Dong et al \cite{DXYCX} prove that $k$-uniform $\MS$ is $\NP$-complete on unit disk graphs for every $k \geq 2$.
We prove a stronger result for $1$-uniform $\MS$.

\begin{theorem}\label{udg-np-complete}
1-uniform $\MS$ is $\NP$-complete on planar unit disk graphs given with a geometric representation.
\end{theorem}

The proof is inspired by Theorem 4.1 in \cite{unit-disk-graphs}.
As in \cite{unit-disk-graphs} we use the following lemma:

\begin{lemma}\label{valiant81}\cite{valiant81}
A planar graph $G$ with maximum degree 4 can be
embedded in the plane using $O(|V|)$ area in such a way that its vertices are at integer coordinates and its edges are drawn so that they are made up of horizontal or vertical segments. 
\end{lemma}

\begin{proof}(of Theorem \ref{udg-np-complete})
We show a reduction from \textsc{PlanarVertexCover} with maximum
degree 3 which is $\NP$-complete \cite{vc-planar-graphs}.
Let $G = (V, E)$ be a planar graph with maximum degree 3.
Let $\{e_1, \ldots, e_{|E|}\}$ be the edges in $G$.
Let $N > 0$ be a sufficient large integer.
We draw $G$ in the plane using Lemma \ref{valiant81} (see Figure \ref{fig-udg-grid}) and we multiply each coordonate by $N$ i.e. each vertex is at coordonate $(iN, jN)$ for some integers $i$ and $j$.
We build $G' = (V', E')$ from $G$ by replacing each edge $e_i = \{u, v\}$ with a subgraph $G_{e_i}$ of vertices
$\{a_{i,0} = u, b_{i,0}, b'_{i,0}, a_{i,1}, b_{i,1}, b'_{i,1}, \ldots, a_{i,2n_i}, b_{i,2n_i}, b'_{i,2n_i}, a_{i,2n_i+1} = v\}$ where each $n_i$ is an integer that depends on the length of the embedding of $e_i$.
For each $i \in [0, 2n_i]$, we connect $b_i$ and $b'_i$ to $a_i$
and $a_{i+1}$ and we connect $b_i$ to $b'_i$
(see Figure \ref{UDG_transformation}).

\begin{figure}[htb]
\centering

\tikzstyle{vertex}=[circle, draw, fill=black!50, inner sep=0pt, minimum width=6pt]
\tikzstyle{tedge}=[very thick,blue]
                        
\begin{tikzpicture}[scale=0.6]

\draw[step=1cm,gray,very thin, dashed] (0,0) grid (6,6);

\node[vertex,label={[xshift=0.2cm,yshift=-0.1cm]$v_1$}] (v1) at (1, 2) {};
\node[vertex,label={[xshift=0.2cm,yshift=-0.1cm]$v_2$}] (v2) at (3, 2) {};
\node[vertex,label={[xshift=0.2cm,yshift=-0.1cm]$v_3$}] (v3) at (5, 2) {};
\node[vertex,label={[xshift=0.2cm,yshift=-0.1cm]$v_4$}] (v4) at (3, 5) {};

\draw [tedge] (v1) -- (v2)  -- (v3); 
\draw [tedge] (v2) -- (v4);
\draw [tedge] (v1) -- (1, 1) -- (1, 1) -- (5, 1) -- (v3);
\draw [tedge] (v1) -- (1, 5) -- (v4);
\draw [tedge] (v3) -- (5, 5) -- (v4);

\draw [->, very thick] (-2, 2.5) -- (-1, 2.5);

\node[vertex,label={[xshift=-0.2cm,yshift=-0.1cm]$v_1$}] (n1) at (-6, 1) {};
\node[vertex,label={[xshift=0.2cm,yshift=-0.1cm]$v_2$}] (n2) at (-3, 1) {};
\node[vertex,label={[xshift=0.2cm,yshift=-0.1cm]$v_3$}] (n3) at (-4.5, 2) {};
\node[vertex,label={[xshift=0.2cm,yshift=-0.1cm]$v_4$}] (n4) at (-4.5, 4) {};
\draw (n1) -- (n2) -- (n3) -- (n1) -- (n4) -- (n2);
\draw (n3) -- (n4);

\end{tikzpicture}
 \caption{A representation of $K_4$ in the grid}
 \label{fig-udg-grid}
\end{figure}
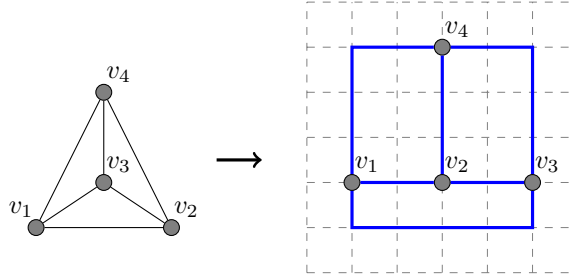

\begin{figure}[htb]
\centering
\begin{tikzpicture}[scale=0.6]
 \node (n1) at (-6,0) {$u$};
 \node (n2) at (-4,0) {$v$};
 \path  (n1) edge (n2);

 \node (u) at (0, 0) {$u$};
 \node (b0) at (2,1) {$b_{i,0}$};
 \node (c0) at (2, -1) {$b'_{i,0}$};
 \node (a1) at (4, 0) {$a_{i,1}$};
 \node (b1) at (6, 1) {$b_{i,1}$};
 \node (c1) at (6, -1) {$b'_{i,1}$};
 \node (a2) at (8, 0) {$a_{i,2}$};
 \node (b2) at (10, 1) {$b_{i,2}$};
 \node (c2) at (10, -1) {$b'_{i,2}$};
 \node (v) at (12, 0) {$v$};
 
 \draw[->,very thick] (-3,0) -- (-2, 0);
 
 \path
 (u) edge (b0) 
 (u) edge (c0)
 (b0) edge (c0)
 (b0) edge (a1)
 (c0) edge (a1)
 (a1) edge (b1)
 (a1) edge (c1)
 (b1) edge (c1)
 (b1) edge (a2)
 (c1) edge (a2)
  (a2) edge (b2) 
 (a2) edge (c2)
 (b2) edge (c2)
 (b2) edge (v)
 (c2) edge (v);
 
 \foreach \x in {u,b0,c0,a1,b1,c1,a2,b2,c2,v} {
   \draw[blue] (\x) circle (1.3);
 }

 \end{tikzpicture}
 \caption{An edge $e_i = \{u, v\}$ and its associate graph $G_{e_i}$ for $n_i=1$}
 \label{UDG_transformation}
\end{figure}
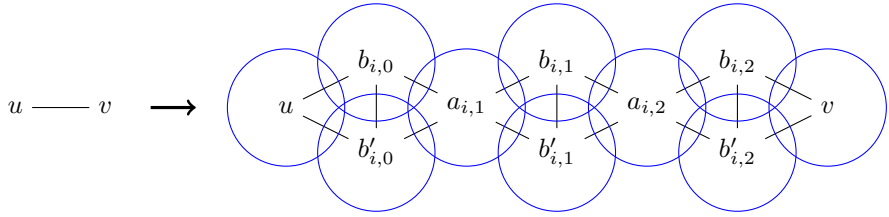


It is easily seen that the obtained graph $G'$ is planar and that there exists an unit disk representation of $G'$ for $N$ sufficient large.
Now, we prove that $G$ admits a vertex-cover $S$ such that $|S| \leq k$ if and only if $G'$ has a monitoring set $S'$
such that $|S'| \leq k' = k + \sum_{i \in [1, |E|} (5n_i+2)$.
Let $A$ be the set of vertices $a_{i,j}$ for $i \in [1, |E|]$ and $j \in [1, 2n_i]$. 
Let $B$ be the set of vertices $b_{i,j}$ and $b'_{i,j}$ for $i \in [1, |E|]$ and $j \in [0, 2n_i]$.
Clearly, $V'$ is the disjoint union of $V$, $A$ and $B$.
Moreover, $|B| = \sum_{e_i \in E} (4n_i+2)$. The proof is an immediate consequence of these three facts.

(1) If a set $S \subseteq V'$ monitors $G'$ then $B \subseteq S$: otherwise, there exists a vertex $b_{i,j}$ or $b'_{i,j}$ that is not in $S$.
Then $\{ a_{i,j}, b_{i, j}\}$ or $\{ a_{i,j}, b'_{i, j}\}$ is not monitored by $S$.

(2) Let $S$ be a vertex-cover of $G$. Then there is a set $A' \subseteq A$ such that $|V(G_{e_i}) \cap A'| = n_i$ for every $i \in [1, |E|]$
and such that $S \cup A' \cup B$ is a monitoring set of $G'$: let $e_i = \{u, v\}$ be an edge in $G$. If $u \in S$, then we choose $a_{2i}$ for $i \in [1, n_i]$
as elements of $A'$. Otherwise ($v \in S$), we choose $a_{2i+1}$ for $i \in [0, n_i-1]$. It is easily seen that $S \cup A' \cup B$ is a monitoring set of $G'$.

(3) There exists a minimum monitoring set $S$ of $G'$ such that $V \cap S$ is a vertex-cover of $G$ and
$|V(G_{e_i} \cap A \cap S| = n_i$ for every $i \in [1, |E|]$: assume that $V \cap S$ is not a vertex cover of $G$.
Let $e_i = \{ u,  v\}$ be an edge in $G$ not covered by $V \cap S$.
Then, it is easily seen that $|V(G_{e_i}) \cap A \cap S| > n_i$.
Otherwise, an edge $\{b_{i, j}, b'_{i, j}\}$ for some $j$ is not covered by $S$.
Thus, we can replace these vertices by $u$ and $n_i$ vertices in $V(G_{e_i} \cap A$
which monitors every edge $\{b_{i,j}, b'_{i,j}\}$. By iterating this processus on every edge in $G$, we obtain a set $S'$ with the desired properties.
Now, assume that $V \cap S$ is a vertex cover of $G$ but there is some $i$ such that $|V(G_{e_i}) \cap A \cap S'| \neq n_i$. It is easily seen
that $|V(G_{e_i}) \cap A \cap S'| < n_i$ implies that an edge $\{b_{i, j}, b'_{i, j}\}$ for some $j$ is not covered by $S'$
and $|V(G_{e_i}) \cap A \cap S'| > n_i$ implies that $S'$ is not minimum.
\end{proof}

\subsection{A PTAS for planar graphs}

Now, we introduce a PTAS for planar graphs and fore more general graph classes: apex-minor-free families of graphs.

An apex graph is a graph $G$ such that for some vertex $v$, $G-v$ is planar.
A minor of a graph $G$ is graph that can be obtained from $G$ by a serie of vertex deletions, edge deletions and edge contractions.
Given a graph $H$, a family of graphs $\mathcal{F}$ is $H$-minor-free if $H$ is not a minor of any graph $G \in \mathcal{F}$.
A family of graphs $\mathcal{F}$ is apex-minor-free if it is $H$-minor-free for some apex graph $H$.
A minor-closed family $\mathcal{F}$ of graphs has bounded local treewidth if there is some function $f$ such that every graph in $\mathcal{F}$ with diameter $d$ has treewidth at most $f(d)$.

In our proof, we use this fundamental property.
\begin{theorem}
\cite{eppstein2000diameter} Let $\mathcal{F}$ be a minor-closed family of graphs. Then $\mathcal{F}$ has bounded local treewidth iff $\mathcal{F}$ is apex-minor-free.
\end{theorem}

We also need, the following result of Baste and al.
\begin{theorem}\label{baste}\cite{parametrizedplanar2016}
$\WMS$ is solvable in time in time $2^{O(\textrm{tw}^2 \log C)} |V|$
where $\mathrm{tw}$ is the treewidth of $G$.
\end{theorem}

Notice that the proof of this theorem in \cite{parametrizedplanar2016} does not consider weights on vertices but we can easily generalize it.

Now, We can prove our theorem.

\begin{theorem}\label{ptas-planar}
There exists a PTAS for $\WMS$ on any weighted apex-minor-free families of graphs.
\end{theorem}

\begin{proof}
Let $H$ be an apex graph.
Without loss of generality, we consider $\WMS$ on the maximal  minor-closed family of graphes $\mathcal{F}$ that excludes the graph $H$.
Thus, $\mathcal{F}$ has bounded local treewidth.
We use the classical Baker's technique \cite{baker94} on planar graphes generalized by Eppstein \cite{eppstein2000diameter} on bounded local treewidth families of graphs. 
Let $(G, w, c)$ be a weighted graph with $G \in \mathcal{F}$.
Without loss of generality, we assume that $G$ is connected.
We choose an arbitrary vertex $v \in V$ and we define $L_i$ as set of vertices at distance $i$ from $v$.
$L_i$ is called the layer of level $i$. Let $l$ be the maximal distance between $v$ and a vertex of $G$.
These layers can be obtained in linear time by breadth first search. 
The key idea is that, since $\mathcal{F}$ has bounded local treewidth, the graph induced by $k$ consecutive layers $L_i, \ldots , L_{i+k-1}$ has a treewidth bounded by $f(k+1)$ \footnote{the supergraph obtained from $G$ by removing all layers $L_j$ with $j \geq i+k$ and by contracting all layers $L_j$
with $j < i$ into one vertex belongs to $\mathcal{F}$ and has a diameter at most $k+1$. Thus its treewidth is at most $f(k+1)$}.
Another important point is that every edge of $G$ has extremities in the same layer or in two consecutive layers.

Fix $\epsilon > 0$ and $k$ such that $\frac{k+2}{k} \leq 1+ \epsilon$.
We will give a $\frac{k+2}{k}$-approximation algorithm that is polynomial for a fixed $k$. We define $B_i$ as the union of $k$ consecutive layers
$L_{i} \cup \ldots \cup L_{i+k-1}$ \footnote{if the index of a layer is not in the interval $[0, l]$, the layer is considered empty} and $R_i$ as 
the union of $k+2$ consecutive layers $L_{i-1} \cup \ldots \cup L_{i+k}$.
Let $P_i$ be the subproblem whose output is a set $S$ of minimum weight in $R_i$ that monitors all edges having at least one extremity in $B_i$.
Since $G[R_i]$ has treewidth at most $f(k+1)$, we can solve this problem in polynomial time using Theorem \ref{baste} by replacing the weight $c(e)$ of edges $e$ having both extremities outside $B_i$ with $0$.

Now, we present Algorithm 3 that is a PTAS for $\WMS$ on  apex-minor-free families of graphs.
\begin{algorithm}
 \label{algo_planar}
\caption{PTAS for $\WMS$ on an apex-minor-free family of graphs}
 \begin{algorithmic}[1] 
\Require $(G, c, w)$, $\epsilon > 0$
\State let $k$ such that $\frac{k+2}{k} \leq 1+ \epsilon$
\If{there exists an edge $e \in E$ such that $c(e) > M(e)$}
\State \Return False
\EndIf
\For{$i$ from $0$ to $k-1$}
   \For{$j$ from $-1$ to $\lceil \frac{l}{k} \rceil$}
   \State let $S_{i,j}$ be an optimal solution of $P_{i+kj}$
   \EndFor
   \State let $S_i= S_{i,-1} \cup \ldots \cup S_{i,\lceil \frac{l}{k} \rceil}$
\EndFor
 \State let $S$ be a set $S_i$ such that $w(S_i)$ is minimal
 \State     \Return $S$

 \end{algorithmic}
\end{algorithm}

It is clear that Algorithm 3 runs in polynomial time when $\epsilon$ is fixed.
Let us prove that Algorithm 3 is correct.
First, notice that there exists a monitoring set of $(G,c)$ if and only if Line 2 of Algorithm 3 fails.
Now, assume that $(G,c)$ admits a monitoring set and let $\mbox{OPT}$ be an optimal solution for $\WMS(G,c,w)$.
Notice that, for any $i$, 
$S_i$ is a (not necessarily optimal) monitoring set of $(G,c)$.
Thus, $S$ is also a monitoring set of $(G,c)$.

Besides, $\mbox{OPT} \cap R_i$ is a (not necessarily optimal) solution of $P_i$.
Indeed, an edge that have an extremity in $B_i$ can only be monitored by vertices in $R_i$.
Consequently, for any $i$ and $j$, it holds that $w(S_{i,j}) \leq w(\mbox{OPT} \cap R_{i+kj})$.
Therefore, for any $i$, we have  $$ w(S_i) \leq  \sum_{j=-1}^{\lceil \frac{l}{k} \rceil}w(S_{i,j}) \leq 
\sum_{j=-1}^{\lceil \frac{l}{k} \rceil} w(\mbox{OPT} \cap R_{i+kj}) $$

There exists an integer $i \in [0,k-1]$ such that $w(\mbox{OPT} \cap (C_i \cup C_{i+1})) \leq \frac{2}{k} w(\mbox{OPT})$
where $C_i$ is the union of layers $L_{i'}$ with $i'$ congruent to $i$ modulo $k$.
Hence, there exists an integer  $i \in [0,k-1]$ such that 
$$  \sum_{j=-1}^{\lceil \frac{l}{k} \rceil} w(\mbox{OPT} \cap R_{i+kj}) \leq \frac{k+2}{k} w(\mbox{OPT}) $$
Thus, we obtain that $w(S) \leq \frac{k+2}{k} w(\mbox{OPT})$.
\end{proof}


\section{Conclusion and Further works}
In this paper, we considered a variant of the dominating set problem, called the edge monitoring problem on several classes of graphes.
We also discussed the weighted version of the edge monitoring problem.
In this section, we list a variety of problems for further work.

\par\leavevmode\par 
\textbf{Problem 1}: Study the problem on other classes of graphes: permutation graphs, strongly chordal graphs, etc.\\

\textbf{Problem 2}: Consider the following variant of the edge monitoring problem: assume that each vertex can monitor only a fixed  number of edges $t$. \\

\textbf{Problem 3}: Consider the variant of the edge monitoring problem where the monitoring set need to be connected,
namely connected edge monitoring problem. 

\bibliographystyle{abbrv}
\bibliography{biblio}

\end{document}